\documentclass[11pt]{article}

\usepackage[margin=1in]{geometry}
\usepackage{amsmath,amssymb,amsthm,mathtools}
\usepackage{microtype}
\usepackage{enumitem}
\usepackage[hidelinks]{hyperref}
\usepackage{cleveref}

\newtheorem{theorem}{Theorem}
\newtheorem{lemma}{Lemma}

\newcommand{\Min}{\operatorname{Min}}
\newcommand{\E}{\mathbb{E}}
\newcommand{\1}{\mathbf{1}}
\newcommand{\width}{\operatorname{width}}
\newcommand{\Rlv}{R^{\mathrm{LV}}_{n,w}}

\title{The Randomized Query Complexity of Finding Minimal Elements in Bounded-Width Posets}
\author{Luyao Fan\textsuperscript{1}, Jiayang Zou\textsuperscript{1,2}, Jiayang Gao\textsuperscript{1}, and Jia Wang\textsuperscript{1}\\[0.5ex]
\small \textsuperscript{1}Shanghai Jiao Tong University, Shanghai, China; \texttt{\char123 fanluyao, qiudao, gjy0515, jiawang\char125 @sjtu.edu.cn}\\
\small \textsuperscript{2}Stanford University, Stanford, CA, USA; \texttt{jyangzou@stanford.edu}}

\date{}

\begin{document}
\maketitle

\begin{abstract}
We study the zero-error randomized query complexity of finding all minimal elements in an unknown $n$-element poset of width at most $w$. Previous work of Daskalakis, Karp, Mossel, Riesenfeld, and Verbin
established a randomized upper bound with leading term
$\frac{w+1}{2}n$, while the corresponding lower bound left a
multiplicative gap in the leading constant that approaches a factor
of 2 as $w$ grows. We prove the finite lower bound
\(
R^{\mathrm{LV}}_{n,w}\ge
\frac{w+1}{2}n-\frac{w(w+3)}4
+w\left(1-\frac1w\right)^n
+\frac{w(w-1)}4\left(1-\frac2w\right)^n.
\)
Consequently, for every fixed $w$,
\(
R^{\mathrm{LV}}_{n,w}
=
\left(\frac{w+1}{2}+o(1)\right)n.
\)
Thus the known randomized upper bound has the correct asymptotic leading constant for every fixed width. The argument is based on a pairwise accounting of incomparable queries under a random-chain hard distribution, using a component-flip involution and a unique ownership property for incomparable comparisons. Generative AI was used in the preparation of this manuscript.
\end{abstract}

\section{Introduction}

Let $P=(V,\prec)$ be a poset.  In the comparison-query model, a query on distinct $x,y\in V$
returns one of
\[
x\prec y,\qquad y\prec x,\qquad x\parallel y.
\]
The task is to determine
\[
\Min(P)=\{x\in V:\nexists y\in V\text{ with }y\prec x\}
\]
under the promise $\width(P)\le w$.

Comparison problems on partially ordered sets go back at least to Faigle and Tur\'an~\cite{FaigleTuran}.
Boldi, Chierichetti, and Vigna studied related top-element extraction problems~\cite{BoldiChierichettiVigna}.
Daskalakis, Karp, Mossel, Riesenfeld, and Verbin studied sorting and selection in bounded-width
posets~\cite{DKMRV}.  For finding all minimal elements they gave a Las Vegas algorithm using
\[
\frac{w+1}{2}n + O_w(\log n)
\]
expected queries for fixed $w$, while their randomized lower bound has
leading term $\frac{w+3}{4}n$~\cite[Sec.~4.2.2]{DKMRV}.
Thus, prior to the present work, the ratio between the upper- and
lower-bound leading constants was
\[
\frac{2(w+1)}{w+3},
\]
which approaches $2$ as $w \to \infty$.

\begin{theorem}\label{thm:main}
For all integers $n\ge1$ and $w\ge2$,
\[
\Rlv\ge
\frac{w+1}{2}n-\frac{w(w+3)}4
+w\left(1-\frac1w\right)^n
+\frac{w(w-1)}4\left(1-\frac2w\right)^n.
\]
Hence, for every fixed $w$,
\[
R^{\mathrm{LV}}_{n,w}
=\left(\frac{w+1}{2}+o(1)\right)n.
\]
\end{theorem}

\section{Model and hard distribution}

Let $V=[n]$.  A deterministic algorithm is \emph{exact} if it terminates and outputs $\Min(P)$
on every legal input.  For a Las Vegas randomized algorithm $\mathcal A$, let $Q_{\mathcal A}(P)$
be its number of queries and define
\[
R^{\mathrm{LV}}_{n,w}
=
\inf_{\mathcal A}
\sup_{\width(P)\le w}
\E Q_{\mathcal A}(P).
\]

Independently assign each $x\in V$ a color
\[
C_x\sim\operatorname{Unif}([w])
\]
and independently choose a uniformly random permutation $\pi$ of $V$.
Write $r(x)$ for the position of $x$ in $\pi$, and define
\begin{equation}\label{eq:hard}
x\prec y
\quad\Longleftrightarrow\quad
C_x=C_y\ \text{ and }\ r(x)<r(y).
\end{equation}
Every realization is a disjoint union of at most $w$ chains, hence has width at most $w$.
Denote this distribution by $\mathcal D_{n,w}$.

Fix a deterministic exact algorithm $A$.
Let $N_=$ and $N_\parallel$ denote the numbers of comparable and incomparable query answers,
respectively, so
\begin{equation}\label{eq:split}
Q_A=N_=+N_\parallel.
\end{equation}

\section{Comparable queries}

\begin{lemma}[Direct witness]\label{lem:witness}
If $x$ is nonminimal in a realization of $\mathcal D_{n,w}$, then before termination the algorithm
queries some pair $\{x,y\}$ and receives the answer $y\prec x$.
\end{lemma}

\begin{proof}
Suppose not.  Let $a=C_x$.  Keep all colors fixed and modify $\pi$ by moving $x$ before every
other color-$a$ element, preserving the relative order of all other vertices.  Then $x$ becomes
minimal.  Every previous query answer is unchanged: queries not involving $x$ are unchanged;
queries between $x$ and another color remain incomparable; a queried relation $x\prec z$ remains
true; and by assumption no queried relation $z\prec x$ exists.  The deterministic algorithm
therefore has the same transcript and output on two legal inputs with different minimal sets,
contradicting exactness.
\end{proof}

Let
\[
K=\bigl|\{c\in[w]:\exists x,\ C_x=c\}\bigr|.
\]
There are $n-K$ nonminimal vertices.  By \Cref{lem:witness}, each is the larger endpoint of at
least one comparable query, and a comparable query has only one larger endpoint.  Hence
\[
N_=\ge n-K.
\]
Since
\[
\E K=w\left(1-\left(1-\frac1w\right)^n\right),
\]
we obtain
\begin{equation}\label{eq:comp}
\E N_=
\ge
n-w+w\left(1-\frac1w\right)^n.
\end{equation}

\section{Incomparable queries}

Fix distinct colors $a,b\in[w]$ and set
\[
S_{ab}=\{x:C_x\in\{a,b\}\},
\qquad
M_{ab}=|S_{ab}|.
\]
Then
\begin{equation}\label{eq:Mbin}
M_{ab}\sim\operatorname{Bin}\left(n,\frac2w\right).
\end{equation}

Condition on $S_{ab}$, on $\pi$, and on all colors $(C_x)_{x\notin S_{ab}}$.
Write $m=|S_{ab}|$ and encode the remaining colors by
\[
\sigma_x=
\begin{cases}
0,&C_x=a,\\
1,&C_x=b.
\end{cases}
\]
The $2^m$ assignments $\sigma\in\{0,1\}^{S_{ab}}$ are equiprobable.  This remains true even
though different assignments may use different numbers of nonempty color classes: the hard
distribution is uniform over labeled color assignments in $[w]^n$, and every such assignment is
legal because the promise is $\width(P)\le w$.

For a fixed execution, define the pair-query graph $G^{ab}$ on vertex set $S_{ab}$ by adding an
edge whenever the algorithm queries two vertices in $S_{ab}$.  For such a query,
\[
x,y\text{ comparable}\iff \sigma_x=\sigma_y,
\qquad
x\parallel y\iff \sigma_x\ne\sigma_y.
\]
A query whose endpoints lie in distinct connected components of the current $G^{ab}$ is called
an \emph{external merge}.

\begin{lemma}[Component-flip pairing]\label{lem:pairing}
Fix a decision-tree node $T$ at which the next query is $\{x,y\}\subseteq S_{ab}$ and $x,y$ lie
in different connected components of the current $G^{ab}$.  Let $\Omega_T$ be the set of binary
assignments reaching $T$.  Then
\[
\bigl|\{\sigma\in\Omega_T:\sigma_x=\sigma_y\}\bigr|
=
\bigl|\{\sigma\in\Omega_T:\sigma_x\ne\sigma_y\}\bigr|.
\]
\end{lemma}

\begin{proof}
Let $C$ be the current connected component containing $x$.  Define
\[
(\Phi\sigma)_z=
\begin{cases}
1-\sigma_z,&z\in C,\\
\sigma_z,&z\notin C.
\end{cases}
\]
Then $\Phi^2=\mathrm{id}$.

Every query made before $T$ has the same answer under $\sigma$ and $\Phi\sigma$.  If both
endpoints lie in $C$, equality or inequality of their labels is preserved, and the orientation
of a comparable answer is preserved because $\pi$ is fixed.  There was no previous query with
one endpoint in $C$ and the other in $S_{ab}\setminus C$, since such a query would connect the
components.  If exactly one endpoint lies in $C$ and the other lies outside $S_{ab}$, the latter
has color outside $\{a,b\}$, so the answer is incomparable before and after the flip.  Queries
disjoint from $C$ are unchanged.  Since $A$ is deterministic, $\Phi\sigma$ reaches the same node
$T$.  Thus $\Phi$ is an involution of $\Omega_T$.

Because $x\in C$ and $y\notin C$,
\[
\sigma_x=\sigma_y
\quad\Longleftrightarrow\quad
(\Phi\sigma)_x\ne(\Phi\sigma)_y.
\]
Hence $\Phi$ bijects the two sets in the statement.
\end{proof}

\begin{lemma}[Terminal connectivity]\label{lem:connectivity}
If $M_{ab}>0$, then the terminal pair-query graph $G^{ab}$ is connected.
\end{lemma}

\begin{proof}
Let $V_a=\{x:C_x=a\}$ and suppose $V_a\ne\varnothing$.  Let $m_a$ be its minimum.
By \Cref{lem:witness}, every $x\in V_a\setminus\{m_a\}$ has a queried witness $y\prec x$ with
$y\in V_a$.  Repeatedly following such witnesses strictly decreases permutation rank and
therefore reaches $m_a$.  Thus $V_a$ is connected in the terminal $G^{ab}$.  The same holds for
$V_b$.

If both $V_a,V_b$ are nonempty and form distinct terminal components, then no query was ever made
between them.  Recolor every vertex of $V_b$ with color $a$, leaving all other colors and $\pi$
unchanged.  Every queried answer is preserved.  The modified input is still legal, while the two
previous chain minima become comparable and only the earlier one in $\pi$ remains minimal.
Thus the terminal transcript would have two different correct outputs, contradicting exactness.
\end{proof}

If $m=M_{ab}>0$, the graph starts with $m$ components and ends with one.  Each external merge
reduces the number of components by one, and no other query does.  Hence every execution has
exactly
\begin{equation}\label{eq:merges}
(m-1)_+
\end{equation}
external merges.

Let $Q_{ab}$ be the number of incomparable external merges for the pair $\{a,b\}$.

\begin{lemma}[Pair charge]\label{lem:paircharge}
For every unordered pair $\{a,b\}$,
\[
\E Q_{ab}=\frac12\E(M_{ab}-1)_+.
\]
\end{lemma}

\begin{proof}
Under the conditioning above, there are $2^m$ equiprobable binary assignments, each having
$(m-1)_+$ external merges by \eqref{eq:merges}.  Thus the total number of external-merge
occurrences is $2^m(m-1)_+$.  Partition these occurrences by decision-tree node.  At every such
node, \Cref{lem:pairing} pairs the compatible assignments yielding a comparable merge with those
yielding an incomparable merge.  Exactly half of all external-merge occurrences are therefore
incomparable.  Averaging over the $2^m$ assignments gives
\[
\E[Q_{ab}\mid S_{ab},\pi,(C_x)_{x\notin S_{ab}}]
=\frac{(m-1)_+}{2}.
\]
Removing the conditioning proves the claim.
\end{proof}

From \eqref{eq:Mbin},
\[
\E M_{ab}=\frac{2n}{w},
\qquad
\Pr(M_{ab}=0)=\left(1-\frac2w\right)^n.
\]
Since $(M-1)_+=M-\1_{\{M>0\}}$,
\begin{equation}\label{eq:pairvalue}
\E Q_{ab}
=
\frac nw-\frac12+\frac12\left(1-\frac2w\right)^n.
\end{equation}

\begin{lemma}[Unique ownership]\label{lem:ownership}
For every execution,
\[
\sum_{1\le a<b\le w}Q_{ab}\le N_\parallel.
\]
\end{lemma}

\begin{proof}
If an incomparable physical query $\{x,y\}$ is counted by $Q_{ab}$, then its true colors
$c=C_x$ and $d=C_y$ satisfy $c\ne d$ and $c,d\in\{a,b\}$.  Hence
\[
\{c,d\}=\{a,b\}.
\]
Thus the query is charged to at most one unordered color pair.
\end{proof}

Taking expectations and using linearity,
\begin{align}
\E N_\parallel
&\ge
\sum_{a<b}\E Q_{ab}\notag\\
&=
\binom w2
\left[
\frac nw-\frac12+\frac12\left(1-\frac2w\right)^n
\right]\notag\\
&=
\frac{w-1}{2}n-\frac{w(w-1)}4
+\frac{w(w-1)}4\left(1-\frac2w\right)^n.
\label{eq:incomp}
\end{align}

\section{Proof of the main theorem}

Combining \eqref{eq:split}, \eqref{eq:comp}, and \eqref{eq:incomp}, every deterministic exact
algorithm satisfies
\begin{align*}
\E_{P\sim\mathcal D_{n,w}}Q_A(P)
&\ge
\frac{w+1}{2}n-\frac{w(w+3)}4
+w\left(1-\frac1w\right)^n\\
&\qquad
+\frac{w(w-1)}4\left(1-\frac2w\right)^n
=:L_{n,w}.
\end{align*}

Let $\mathcal A$ be a Las Vegas randomized algorithm with internal randomness $R$, and let $A_r$
be the deterministic algorithm obtained by fixing $R=r$. Write $Q(P,r)$ for the number of queries made by $A_r$ on input $P$. By the Las Vegas property, for every legal input $P$,
\[
\Pr_R\!\left(
A_R(P)\text{ terminates and outputs }\operatorname{Min}(P)
\right)=1.
\]
Since the set of legal posets on the fixed labeled ground set is finite,
outside a single null set of random seeds, $A_r$ terminates and is
correct on every legal input. Thus $A_r$ is exact for almost every $r$. Hence for almost every $r$,
\[
\E_{P\sim\mathcal D_{n,w}}Q(P,r)\ge L_{n,w}.
\]
Averaging over $R$ and using Tonelli's theorem,
\[
\E_{P\sim\mathcal D_{n,w}}\E_RQ(P,R)\ge L_{n,w}.
\]
Thus some legal $P^*$ satisfies $\E_RQ(P^*,R)\ge L_{n,w}$, and therefore
\[
\sup_{\width(P)\le w}\E_RQ(P,R)\ge L_{n,w}.
\]
Taking the infimum over Las Vegas algorithms proves the lower bound in \Cref{thm:main}.

For fixed $w$,
\[
R^{\mathrm{LV}}_{n,w}
\ge
\frac{w+1}{2}n-O_w(1).
\]
Daskalakis et al.~\cite{DKMRV} give
\[
R^{\mathrm{LV}}_{n,w}
\le
\frac{w+1}{2}n+O_w(\log n).
\]
Hence
\[
\lim_{n\to\infty}\frac{R^{\mathrm{LV}}_{n,w}}n
=
\frac{w+1}{2}.
\]
\qed

\section*{AI Assistance}
Generative AI was used in the preparation of this manuscript.


\begin{thebibliography}{9}

\bibitem{FaigleTuran}
U.~Faigle and Gy.~Tur\'an,
\newblock Sorting and recognition problems for ordered sets,
\newblock \emph{SIAM Journal on Computing} 17(1):100--113, 1988.

\bibitem{BoldiChierichettiVigna}
P.~Boldi, F.~Chierichetti, and S.~Vigna,
\newblock Pictures from Mongolia: Extracting the top elements from a partially ordered set,
\newblock \emph{Theory of Computing Systems} 44(2):269--288, 2009.

\bibitem{DKMRV}
C.~Daskalakis, R.~M. Karp, E.~Mossel, S.~J. Riesenfeld, and E.~Verbin,
\newblock Sorting and selection in posets,
\newblock \emph{SIAM Journal on Computing} 40(3):597--622, 2011.
\newblock Preliminary version in SODA 2009; \href{https://arxiv.org/abs/0707.1532}{arXiv:0707.1532}.

\end{thebibliography}
\end{document}